\documentclass[showpacs,amsmath,amssymb,10pt,aps]{revtex4}
  
\usepackage[utf8]{inputenc}
\usepackage{amsmath}
\usepackage{amsfonts}
\usepackage{graphicx}
\usepackage{yfonts}
\usepackage{color}
\usepackage[normalem]{ulem}
\usepackage{amsthm}
\usepackage{bm}
\usepackage{bbm}
\usepackage{mathtools}

\DeclareMathOperator{\imag}{i}

\newcommand\bigforall{\mbox{\Large $\mathsurround=1pt\forall$}}

\def\<{\langle}
\def\>{\rangle}

\newtheorem{example}{Example}
\usepackage{placeins}

\newcommand{\ot}{{\,\otimes\,}}

\newcommand{{\Cd}}{{\mathbb{C}^d}}

\newcommand{\Tr}{\mathrm{Tr}}

\def\oper{{\mathchoice{\rm 1\mskip-4mu l}{\rm 1\mskip-4mu l}
{\rm 1\mskip-4.5mu l}{\rm 1\mskip-5mu l}}}
\def\<{\langle}
\def\>{\rangle}
\newtheorem{Theorem}{Theorem}

\newtheorem{Corollary}{Corollary}

\newtheorem{Remark}{Remark}
\newtheorem{Proposition}{Proposition}

\newcommand{\beq}{\begin{equation}}
\newcommand{\eeq}{\end{equation}}
\newcommand{\bear}{\begin{eqnarray}}
\newcommand{\ear}{\end{eqnarray}}
\newcommand{\bdm}{\begin{displaymath}}
\newcommand{\edm}{\end{displaymath}}

\begin{document}

\title{\bf Quantum channels irreducibly covariant with respect to the finite group \\ generated by the Weyl operators}
\author{ Katarzyna Siudzi{\'n}ska and Dariusz Chru{\'s}ci{\'n}ski }
\affiliation{ Institute of Physics, Faculty of Physics, Astronomy and Informatics \\  Nicolaus Copernicus University,
Grudzi\c{a}dzka 5/7, 87--100 Toru\'n, Poland}

\begin{abstract}
We introduce a class of linear maps irreducibly covariant with respect to the finite group generated by the Weyl operators. This group provides a direct generalization of the quaternion group. In particular, we analyze the irreducibly covariant quantum channels; that is, the completely positive and trace-preserving linear maps. Interestingly, imposing additional symmetries leads to the so-called generalized Pauli channels, which were recently considered in the context of the non-Markovian quantum evolution. Finally, we provide examples of irreducibly covariant positive but not necessarily completely positive maps.
%
% that corresponds to the Weyl channels. They are constructed from the group whose generators are the Weyl operators. We provide the conditions for complete positivity and preserving the trace, which follow from the properties of the Choi matrix. Finally, we propose the set of conditions that have to be imposed on the resulting Weyl channels in order to obtain the Weyl channels with real eigenvalues and the generalized Pauli channels.
\end{abstract}

\maketitle

\section{Introduction}

A linear map $\Phi : M_d(\mathbb{C}) \rightarrow M_d(\mathbb{C})$ is covariant with respect to the unitary $d$-dimensional representation $U$ of the group $G$ if and only if
\begin{equation}\label{cov_def}
\bigforall_{g\in G}\bigforall_{X\in M_d(\mathbb{C})}\quad \Phi\big[U(g)X U^\dagger(g)\big]=U(g)\Phi[X]U^\dagger(g).
\end{equation}
If $\Phi$ is also completely positive and trace-preserving, then it is called the {\it covariant quantum channel}. This notion first appeared in \cite{Scutaru}, where Scutaru proved the Stinespring type theorem for covariant completely positive maps. However, the covariant quantum channels were first analyzed together with the covariant Markovian generators by Holevo \cite{Holevo1993,CQME}. Examples of covariant channels include depolarizing channels \cite{Keyl} and transpose depolarizing channels \cite{additiv,DHS}.

The channels that are covariant with respect to the irreducible unitary representation are called irreducibly covariant and posess an interesting property. Namely, if the channel is irreducibly covariant with respect to a compact group, then its Holevo capacity is additive and linearly proportional to the minimum output entropy \cite{Holevo2000}. This discovery started a series of articles dealing with additivity and multiplicity of minimum output entropy for irreducibly covariant quantum channels \cite{additiv,WernerHolevo,DHS2,Holevo2005,DFH,Fukuda,King}.

There are the channels known as Weyl-covariant which are covariant with respect to the Weyl operators \cite{CQME,Amosov,Holevo2005,DFH}.
It turns out that they satisfy condition (\ref{cov_def}) with $G$ being the group generated by the Weyl operators. The properties of unitarily covariant and Weyl-covariant channels in dimension $d\geq 2$ were studied by Datta, Fukuda, and Holevo in \cite{DFH}, whereas Mendl and Wolf \cite{20} dealt with the channels covariant with respect to O($d$). Nuwairan \cite{23} introduced the so-called EPOSIC channels and showed that they form a set of extreme points of the SU(2)-irreducibly covariant channels. Finally, Jen\u{c}ov\'{a} and Pl\'{a}vala \cite{17} provided the optimality conditions in discriminating covariant quantum channels.

In this paper, we construct the linear maps $\Phi$ that are covariant with respect to the unitary representation $U$ of a finite group $G$. In particular, we construct the Weyl maps by requiring that the group $G$ is the generalization of the quaternion group generated by the Weyl operators. It turns out that such $\Phi$ belong to the class of irreducibly covariant linear maps with simply reducible $U\otimes U^c$, which has been recently analyzed by Mozrzymas et. al. \cite{MSD}. We provide the necessary and sufficient conditions for the Weyl maps to be completely positive and trace-preserving. For the $d$-dimensional Weyl channels, there are $d-1$ choices of inequivalent representations $U$. The covariant maps possessing the additional symmetry $S \Phi[X]S^\dagger =\Phi[SXS^\dagger]$, with $S$ being a suitable permutation matrix, satisfy $\Phi=\Phi^\ast$, where $\Phi^\ast$ is the map dual to $\Phi$. Applying the additional conditions that follow from the group representation theory allows us to obtain the generalized Pauli channels for prime dimensions $d$.
Finally, we provide two methods of constructing the Weyl-covariant positive, trace-preserving maps.

\section{Generalization of the quaternion group}

Consider a set of unitary Weyl operators in $\mathbb{C}^d$
\begin{equation}
W_{kl}=\sum_{m=0}^{d-1}\omega^{km}|m+l\>\<m|,\qquad \omega=e^{2\pi\imag/d},
\end{equation}
satisfying well known properties
\begin{equation}\label{group_action}
W_{kl}W_{rs}=\omega^{ks}W_{k+r,l+s},\qquad W_{kl}^\dagger=\omega^{kl}W_{-k,-l}.
\end{equation}
Moreover
\begin{equation}\label{}
\Tr(W^\dagger_{kl} W_{mn}) = \delta_{km}\delta_{ln} ,
\end{equation}
that is, `renormalized' operators $W_{kl}/\sqrt{d}$ define an  orthonormal basis in ${M}_d(\mathbb{C})$ (with respect to the Hilbert-Schmidt inner product).
Let $G$ denote a finite group generated by $W_{kl}$. It consists in the following elements
\begin{equation}
G= \Big\{\,  \omega^mW_{kl}\, |\, k,l,m=0,1,\dots,d-1\, \Big\},
\end{equation}
and hence $|G|=d^3$. One easily  find the conjugacy classes of $G$,
\begin{equation}\label{conj}
\mathcal{C}_{0}^l=\{\omega^lW_{00}\},\quad \mathcal{C}_{kl}=\{\omega^mW_{kl}|m=0,1,\dots,d-1\},
\end{equation}
for $k,l=0,1,\dots,d-1$ and $(k,l)\neq (0,0)$.
\begin{Proposition}
There are `$d(d+1)-1$' irreducible representations of $G$

\begin{enumerate}

\item $d^2$ one-dimensional (denoted by $\phi_0=\mathrm{id},\phi_1,\dots,\phi_{d^2-1}$), and

\item $d-1$ $d$-dimensional unitary representations (denoted by $U_\alpha$, $\alpha=1,\dots,d-1$).

\end{enumerate}

\end{Proposition}
\begin{proof}
The dimensions of all the non-equivalent irreducible representations are related to the number of elements $|G|=d^3$ by
\begin{equation}\label{AAA}
|G|=\sum_{k=1}^Nn_kd_k^2,
\end{equation}
where $n_k$ is the number of $d_k$-dimensional representations. The number of nonequivalent irreducible representations is equal to the number of equivalence classes (\ref{conj}) of the group, of which there are
\begin{equation}\label{BBB}
d+d^2-1=\sum_{k=1}^Nn_k.
\end{equation}
Moreover, we know that every group has the identity representation, which is one-dimensional, so that $d_1=1$. The following property of the characters,
\begin{equation}
|G|=\sum_{g\in G}\left|\chi^\alpha(g)\right|^2,
\end{equation}
and the fact that one-dimensional representations satisfy the group action imply that $\chi^\alpha(g)\in\{1,\omega,\dots,\omega^{d-1}\}$. Moreover, there are at most $n_1=d^2$ orthogonal rows of the character table.
The group $G$ is generated by the $d$-dimensional Weyl operators, and therefore $W_{kl}$ define an irreducible unitary representation of $G$. Denote its dimension by $d_2=d$. From (\ref{AAA}) and (\ref{BBB}), it follows that
\begin{equation}
d-1-n_2=\sum_{k=3}^Nn_k,\qquad\sum_{k=3}^Nn_k\left(1-\frac{d_k^2}{d^2}\right)=0.
\end{equation}
As $d_k\neq d$ for $k\geq 3$, the only solution is $n_k=0$, $k\geq 3$, and $n_2=d-1$.
\end{proof}
The character table has the structure presented below.

\FloatBarrier
\begin{figure}[ht!]
  \centering
\tiny
       \includegraphics[width=0.4\textwidth]{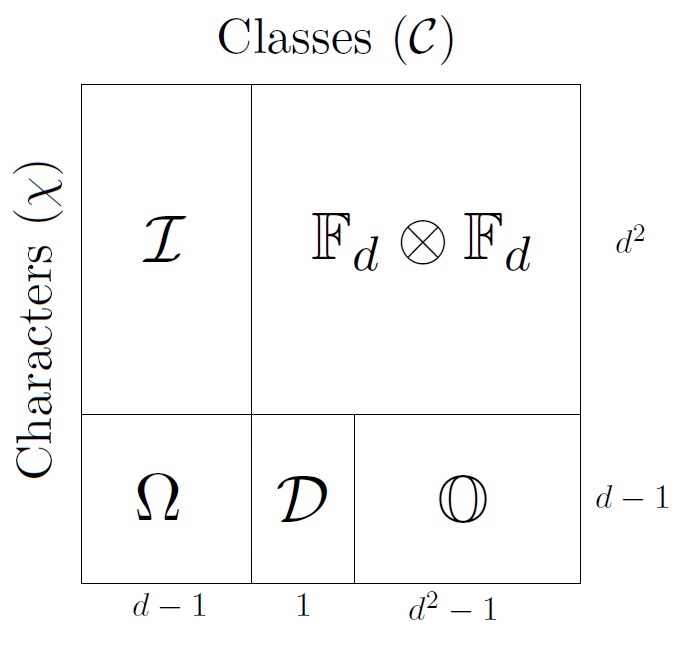}
\end{figure}
\FloatBarrier

In the above table, we introduced the matrices with the following entries:
$\mathcal{I}_{ij}=1$, the Hadamard matrix $(\mathbb{F}_d)_{kl}=\omega^{kl}$ ($k,l=0,\dots,d-1$), $\mathcal{D}_m=d$, $\mathbb{O}_{mn}=0$, and $\Omega_{ij}=d\omega^{ij}$ ($i,j=1,\dots,d-1$).
The characters of the irreducible representations
\begin{equation}
\chi^{\mathrm{id}},\chi^{\phi_1},\dots,\chi^{\phi_{d^2-1}},\chi^{U_1},\dots,\chi^{U_{d-1}}
\end{equation}
number the rows of the table, and the conjugacy classes
\begin{equation}\label{klasy}
\mathcal{C}_{0}^1,\dots,\mathcal{C}_{0}^{d-1},
\mathcal{C}_{0}^0\equiv\mathcal{C}_{00},\dots,\mathcal{C}_{0,d-1},
\dots,
\mathcal{C}_{d-1,0},\dots,\mathcal{C}_{d-1,d-1}
\end{equation}
number the columns.
From definition, the columns and rows of the matrix $\mathbb{F}_d\otimes\overline{\mathbb{F}}_d$ are mutually orthogonal. We check that the one-dimensional representations satisfy the group action (\ref{group_action}),
\begin{equation}
\chi^{md+n}(\mathcal{C}_{kl})\chi^{md+n}(\mathcal{C}_{rs})
=\chi^{md+n}(\mathcal{C}_{k+r,l+s}),
\end{equation}
$m,n=0,\dots,d-1$, as
\begin{equation}
\chi^{md+n}(\mathcal{C}_{kl})=\left(\mathbb{F}_d\otimes\overline{\mathbb{F}}_d \right)_{mn,kl}=\omega^{mk-nl}.
\end{equation}

Let us denote the known unitary representation (with the Weyl operators) by $U_1$. For $d=2$, this is the only irreducible representation of dimension higher than one.
Observe that for prime $d>2$, the other unitary representations follow from the group action written in terms of the group generators,
\begin{equation}\label{action}
W_{kl}W_{rs}=(\omega W_{00})^{ks}W_{k+r,l+s}.
\end{equation}

\begin{Proposition}\label{prop}
For prime $d$, the set of $d$-dimensional irreducible unitary representations of $G$ consist in $U_1$
and the representations constructed from $U_1$ in the following way,
\begin{align}
U_1\to U_\alpha:\quad &\omega^mW_{kl}\to \omega^mW_{\alpha k,\alpha l},\nonumber\\
&\omega^m W_{00}\to \omega^{m\alpha^2}W_{00},\label{U}\\
U_1\to \overline{U}_\alpha:\quad &\omega^mW_{kl}\to \omega^m\overline{W}_{\alpha k,\alpha l}=\omega^mW_{-\alpha k,\alpha l},\nonumber\\
&\omega^m W_{00}\to \omega^{-m\alpha^2}W_{00},\label{Ubar}
\end{align}
for $\alpha=1,\dots,\frac{d-1}{2}$. To number the rows of the character table, they are ordered as follows,
\begin{equation}
U_1,\dots,U_{\frac{d-1}{2}},\overline{U}_{\frac{d-1}{2}},\dots,\overline{U}_1.
\end{equation}
\end{Proposition}

\begin{proof}
The fact that $U_1$ is an irreducible representation of $G$ is obvious, as $U_1(g)$ have the right dimensions, characters, and they satisfy the group action.
After making substitutions (\ref{U}) and (\ref{Ubar}) in eq. (\ref{action}), we get
\begin{equation}
W_{\alpha k,\alpha l}W_{\alpha r,\alpha s}=(\omega^{\alpha^2} W_{00})^{ks}W_{\alpha(k+r),\alpha(l+s)},
\end{equation}
\begin{equation}
\overline{W}_{\alpha k,\alpha l}\overline{W}_{\alpha r,\alpha s}=(\omega^{-\alpha^2} W_{00})^{ks}\overline{W}_{\alpha(k+r),\alpha(l+s)},
\end{equation}
respectively, which agree with (\ref{action}). Observe that $\overline{U}_\alpha$ is the contragredient (dual) representation of $U_\alpha$.

It is worth noting that for $\alpha=\frac{d+1}{2},\dots,d-1$, (\ref{U}) and (\ref{Ubar}) also produce the irreducible unitary representations, but they are equivalent to the ones already obtained (they correspond to the same rows of the character table).
\end{proof}

\section{Covariant quantum channels}

{Consider an arbitrary finite group} $G$ and its unitary  representation $U$ in $\mathbb{C}^d$. Let us construct the linear map
\begin{equation}\label{kanal}
\Phi=\sum_{g\in G}\mu(g)\mathrm{Ad}_{U(g)},
\end{equation}
where the adjoint representation is defined by
\begin{equation}
\mathrm{Ad}_{U(g)}(X)=U(g)XU^\dagger(g).
\end{equation}
This class of maps defines a $\star$-algebra: if
\begin{equation}
\Psi=\sum_{g\in G}\nu(g)\mathrm{Ad}_{U(g)},
\end{equation}
then
\begin{equation}\label{}
  \Phi \Psi = \sum_{g,h\in G}\mu(g)\nu(h)\mathrm{Ad}_{U(gh)} = \sum_{g\in G}(\mu\ast \nu)(g) \mathrm{Ad}_{U(g)},
\end{equation}
where
\begin{equation}\label{}
  (\mu\ast \nu)(g) = \sum_{h \in G} \mu(h) \nu(h^{-1}g) ,
\end{equation}
is the convolution on $G$. Moreover, the dual {map} $\Phi^\ast$ defined by
\begin{equation}\label{}
  {\rm Tr}(\Phi[X]^\dagger Y) =: {\rm Tr}(X^\dagger \Phi^*[Y]) ,
\end{equation}
reads
\begin{equation}
\Phi^* =\sum_{g\in G}\overline{\mu(g)}\mathrm{Ad}_{U(g)}.
\end{equation}

Now, formula (\ref{cov_def}) implies that $\Phi$ is covariant with respect to $U$ if and only if
\begin{equation}
\sum_{h\in G}\mu(h)\mathrm{Ad}_{U(g)U(h)}
=\sum_{h\in G}\mu(h)\mathrm{Ad}_{U(h)U(g)} .
\end{equation}
Using $U(g)U(h)=U(gh)$, one has
\begin{equation}\label{1}
\sum_{h\in G}\mu(h)\mathrm{Ad}_{U(gh)}=\sum_{h\in G}\mu(h)\mathrm{Ad}_{U(hg)}.
\end{equation}
Note that
\begin{equation}
\sum_{h\in G}\mu(h)\mathrm{Ad}_{U(hg)}=\sum_{ghg^{-1}\in G}\mu(ghg^{-1})\mathrm{Ad}_{U(gh)},
\end{equation}
which, together with (\ref{1}), results in the following condition for $\mu(g)$,
\begin{equation}
\bigforall_{g,h\in G}\quad \mu(g)=\mu(hgh^{-1}).
\end{equation}
Therefore, $\mu(g)$ depends on the classes $\mathcal{C}(g)$ of the group elements $g$ instead of the elements themselves.
{Denoting the conjugacy classes by $C_k$ $(k=1,2,\ldots,r)$,} one has
\begin{equation}\label{}
\Phi=\sum_{k=1}^r \mu_k \Phi_k ,
\end{equation}
where $\mu_k=\mu(g)$ for $g \in C_k$, 
and
\begin{equation}\label{}
\Phi_k=\sum_{g\in C_k}\mathrm{Ad}_{U(g)}.
\end{equation}
For the linear map $\Phi$ to describe the quantum channel, additional constraints have to be imposed upon $\mu(g)$.

\begin{Theorem}\label{tw}
The map $\Phi$ in (\ref{kanal}) is the quantum channel if its coefficients $\mu(g)$ satisfy
\begin{equation}\label{warunki}
\mu(g)\geq 0,\qquad \sum_{g\in G}\mu(g)=1.
\end{equation}
\end{Theorem}
Interestingly, as we shall see, conditions (\ref{warunki}) are only sufficient and not necessary.

{Now, consider} the group $G$ generated by the $d$-dimensional Weyl operators. Construct the linear map $\Phi$ that is covariant with respect to the representation $U_1$. Let us assume that $\mu(g)$ is the linear combination of the characters,
\begin{equation}\label{mu}
\mu(\mathcal{C}(g))=\frac{1}{|G|}\left(\sum_{\alpha=0}^{d^2-1}\nu_\alpha\chi^\alpha(\mathcal{C}(g))
+\sum_{\alpha=1}^{d-1}\tau_\alpha\chi^{U_\alpha}(\mathcal{C}(g))\right),
\end{equation}
where $\chi^\alpha$ is the character of the one-dimensional representation $\phi_\alpha$, and $\chi^{U_\alpha}$ denotes the character of the $d$-dimensional representation $U_\alpha$. Using the character table of $G$, one arrives at
\begin{equation}\label{mu1}
\mu(\mathcal{C}_{kl})=\frac{1}{|G|}\sum_{m,n=0}^{d-1}\nu_{mn}\omega^{mk-nl},
\qquad (k,l)\neq (0,0),
\end{equation}
\begin{equation}\label{mu2}
\mu(\mathcal{C}_0^l)=\frac{1}{|G|}\sum_{m,n=0}^{d-1}\nu_{mn}+
\frac{d}{|G|}\sum_{\alpha=1}^{d-1}\tau_\alpha \omega^{\alpha l},
\qquad l=0,\dots,d-1.
\end{equation}
Therefore, the resulting map is the Weyl map
\begin{equation}\label{inv}
\Phi[X]=\sum_{k,l=0}^{d-1}\mu_{kl}W_{kl}XW_{kl}^\dagger,
\end{equation}
where $\mu_{kl}=d\mu(\mathcal{C}_{kl})$ and
\begin{equation}
\mu_{00}=\sum_{l=0}^{d-1}\mu(\mathcal{C}_0^l)=\frac{d}{|G|}\sum_{m,n=0}^{d-1}\nu_{mn}.
\end{equation}
Note that parameters $\tau_\alpha$ are completely redundant, {as} they do not enter the final formula for $\mu_{kl}$.

\begin{example}
For $d=2$, $G=Q_8$ is the quaternion group, for which $U(g)$ are given by the Pauli matrices $\{\pm\sigma_0,\pm\sigma_1,\pm\imag\sigma_2,\pm\sigma_3\}$. Now, we construct the map
\begin{equation}
\Phi[X]=\sum_{k=1}^2\mu(\mathcal{C}_0^k)X+\sum_{k=1}^3\mu(\mathcal{C}_k)\sigma_kX\sigma_k,
\end{equation}
where $\mathcal{C}_1=\{\pm\sigma_1\}$, $\mathcal{C}_2=\{\pm\imag\sigma_2\}$, $\mathcal{C}_3=\{\pm\sigma_3\}$, and $\mathcal{C}_0^k=\{(-1)^k\sigma_0\}$. The coefficients $\mu(\mathcal{C}_0^k)$ and $\mu(\mathcal{C}_k)$ read
\begin{equation}
\begin{split}
\mu(\mathcal{C}_1)&=\frac 18 \left(\nu_{00}-\nu_{01}+\nu_{10}-\nu_{11}\right),\\
\mu(\mathcal{C}_2)&=\frac 18 \left(\nu_{00}-\nu_{01}-\nu_{10}+\nu_{11}\right),\\
\mu(\mathcal{C}_3)&=\frac 18 \left(\nu_{00}+\nu_{01}-\nu_{10}-\nu_{11}\right),\\
\mu(\mathcal{C}_0^0)&=\frac 18 \left(\nu_{00}+\nu_{01}+\nu_{10}+\nu_{11}+2\tau_1\right),\\
\mu(\mathcal{C}_0^1)&=\frac 18 \left(\nu_{00}+\nu_{01}+\nu_{10}+\nu_{11}-2\tau_1\right) ,
\end{split}
\end{equation}
and hence $\Phi$ can be rewritten as
\begin{equation}\label{Pauli}
\Phi[X]=\sum_{k=0}^3\mu_k\sigma_kX\sigma_k
\end{equation}
with $\mu(\mathcal{C}_k)=2\mu_k$, $k=1,\ 2,\ 3$, and
\begin{equation}
\mu_0=\frac 14 \left(\nu_{00}+\nu_{01}+\nu_{10}+\nu_{11}\right).
\end{equation}
Note that $\mu(\mathcal{C}_0^k)$ can be negative, provided that $\mu(\mathcal{C}_0^0) + \mu(\mathcal{C}_0^1) \geq 0$.
\end{example}

\begin{Theorem}\label{quantum_channel}
The Weyl map $\Phi$ given in (\ref{inv}) is a quantum channel if and only if
\begin{equation}\label{CP_war}
\mu_{kl}\geq 0,\qquad k,l=0,\dots,d-1,\qquad
\sum_{k,l=0}^{d-1}\mu_{kl}=1.
\end{equation}
\end{Theorem}

\begin{proof}
Recall that $\Phi$ is a quantum channel if and only if it is completely positive and trace-preserving.
Complete positivity is guaranteed by $J(\Phi)\geq 0$, where $J(\Phi)$ is the associated Choi matrix
\begin{equation}
J(\Phi)=\sum_{i,j=0}^{d-1}e_{ij}\otimes\Phi[e_{ij}],
\end{equation}
and $e_{ij}=|i\>\<j|$. It is easy to see that the Choi matrix satisfies the eigenvalue equation
\begin{equation}
J(\Phi)|v_{kl}\>=d\mu_{kl}|v_{kl}\>
\end{equation}
with
\begin{equation}
|v_{kl}\>=\sum_{i=0}^{d-1}|i\>\otimes W_{kl}|i\>.
\end{equation}
Therefore, $J(\Phi)\geq 0$ is equivalent to $\mu_{kl}\geq 0$. Showing that $\Phi$ preserves the trace if and only if $\sum_{k,l=0}^{d-1}\mu_{kl}=1$ is elementary.
\end{proof}

\section{Construction of the irreducibly covariant quantum channels} \label{IC}

In this section, we follow the method of constructing the irreducibly covariant linear maps presented in \cite{MSD}. Take a finite group $G$ and its irreducible unitary representation $U$. With $U$, we associate the contragredient representation
\begin{equation}\label{contragredient}
U^c(g)=U^T(g^{-1})\equiv\overline{U}(g).
\end{equation}
It turns out that the representation $U\otimes U^c$ is reducible, and therefore it can be expressed as the sum
\begin{equation}\label{UotimesUc}
U\otimes U^c=\bigoplus_{\alpha\in\Theta} m_\alpha \phi_\alpha,
\end{equation}
of irreducible unitary representations $\phi_\alpha$ with multiplicity
\begin{equation}\label{m}
m_\alpha=\frac{1}{|G|}\sum_{g\in G}\chi^\alpha(g^{-1})|\chi^U(g)|^2,
\end{equation}
where $\chi^\alpha$ is the character of the representation $\phi_\alpha$.
The procedure in \cite{MSD} is applicable only to the multiplicity-free case, in which $m_\alpha=1$ for $\alpha\in\Theta$ and $m_\alpha=0$ for $\alpha\notin\Theta$. Then, $U\otimes U^c$ is called simply reducible. The linear map $\Phi$ that is irreducibly covariant with respect to the group $G$ is constructed in terms of its spectral decomposition,
\begin{equation}\label{Phi}
\Phi=\sum_{\alpha\in\Theta}\ell_\alpha\Pi_\alpha.
\end{equation}
In the above formula, $\ell_\alpha$ are the eigenvalues of $\Phi$, and
\begin{equation}\label{Pi}
\Pi_\alpha[X]=\frac{d_\alpha}{|G|}\sum_{g\in G}\chi^\alpha(g^{-1})
U(g)XU^\dagger(g),
\end{equation}
where $|G|$ denotes the order of $G$, $d_\alpha=\dim\phi_\alpha$. Additionally, $\Phi$ is a quantum channel if and only if it is completely positive and trace-preserving.

\begin{Proposition}\label{speccase}
The irreducibly covariant quantum channel (\ref{Phi}) is a special case of the covariant quantum channel in Theorem \ref{tw}, where $U$ is the irreducible unitary representation, $U\otimes U^c$ is simply reducible, and
\begin{equation}
\mu(g)=\frac{1}{|G|}\sum_{\alpha\in\Theta}d_\alpha\ell_\alpha\overline{\chi}^\alpha(g).
\end{equation}
\end{Proposition}

The following theorem shows that the Weyl channels belong to the class presented in Proposition \ref{speccase}.

\begin{Theorem}
For every $U_\alpha$ and $\overline{U}_\alpha$, decomposition (\ref{UotimesUc}) is multiplicity-free -- that is, the only non-vanishing multiplicities correspond to $\phi_\alpha$ and equal $m_\alpha=1$.
\end{Theorem}

\begin{proof}
From the character table, we know that
\begin{equation}
\left|\chi_\beta^{U_\alpha}(\mathcal{C}_{kl})\right|^2=
\left|\chi_\beta^{\overline{U}_\alpha}(\mathcal{C}_{kl})\right|^2=
0
\end{equation}
for $(k,l)\neq(0,0)$ and
\begin{equation}
\left|\chi_\beta^{U_\alpha}(\mathcal{C}_{0}^m)\right|^2=
\left|\chi_\beta^{\overline{U}_\alpha}(\mathcal{C}_{0}^m)\right|^2=
d^2
\end{equation}
for $m=0,\dots,d-1$.
This allows us to simplify eq. (\ref{m}) to
\begin{equation}
m_\alpha=\frac 1d \sum_{g\in\bigcup_{m=0}^{d-1}\mathcal{C}_0^m}\chi^\alpha(g^{-1})=1
\end{equation}
for the one-dimensional representations $\phi_\alpha$ and
\begin{equation}
m_{U_\alpha}=\frac 1d \sum_{g\in\bigcup_{m=0}^{d-1}\mathcal{C}_0^m}\chi_\beta^{U_\alpha}(g^{-1})=
\sum_{k=0}^{d-1}\omega^k=0
\end{equation}
for the $d$-dimensional representations $U_\alpha$. Analogical calculations show that $m_{\overline{U}_\alpha}=0$.
\end{proof}

Now, using the unitary representation $U_1$ of the group $G$, we construct the irreducibly covariant linear map
\begin{equation}\label{phi}
\Phi=\sum_{k,l=0}^{d-1}\ell_{kl}\Pi_{kl},
\end{equation}
where, from definition (\ref{Pi}),
\begin{equation}\label{proj}
\Pi_{kl}[X]=\frac{1}{d^2}\sum_{m,n=0}^{d-1}\omega^{-mk+nl}W_{mn}XW_{mn}^\dagger.
\end{equation}

\begin{Proposition}
{The operators given in (\ref{proj}) are the rank-1 projectors onto the Weyl operators,}
\begin{equation}\label{proj2}
\Pi_{kl}[X]=\frac{1}{d}W_{kl}\Tr(W_{kl}^\dagger X).
\end{equation}
\end{Proposition}

\begin{proof}
{To prove this,} it is enough to see that, when acting on the unitary basis of the Weyl operators, eqs. (\ref{proj}-\ref{proj2}) both produce
\begin{equation}
\Pi_{kl}[W_{mn}]=\delta_{km}\delta_{nl}W_{mn}.
\end{equation}
\end{proof}

\begin{Theorem}
The irreducibly covariant linear map defined in (\ref{phi}) with $\Pi_{kl}$ given by (\ref{proj}) is the Weyl map -- that is,
\begin{equation}\label{WeylCh}
\Phi[X]=\sum_{k,l=0}^{d-1}p_{kl}W_{kl}XW_{kl}^\dagger.
\end{equation}
Moreover, it is the Weyl channel if and only if $p_{kl}\geq 0$ and $\sum_{k,l=0}^{d-1}p_{kl}=1$.
\end{Theorem}

\begin{proof}
If we substitute (\ref{proj}) into (\ref{phi}), we see that
\begin{equation}
\Phi[X]=\frac{1}{d^2}\sum_{k,l,m,n=0}^{d-1}\omega^{-mk+nl}\ell_{kl}W_{mn}XW_{mn}^\dagger,
\end{equation}
which for
\begin{equation}\label{p}
p_{mn}:=\frac{1}{d^2}\sum_{k,l=0}^{d-1}\omega^{-mk+nl}\ell_{kl}
\end{equation}
is exactly (\ref{WeylCh}). Complete positivity and preserving the trace follow from Theorem \ref{quantum_channel} for $\ell_{kl}=\overline{\nu}_{kl}$ and $p_{kl}=\mu_{kl}$.
\end{proof}

\section{Generalized Pauli channels}

For a prime $d$, there exists a subclass of the Weyl channels with additional symmetries, known as the generalized Pauli channels \cite{Ruskai,DCKS,DCKS2}.
{The} Weyl channel $\Phi$ in (\ref{WeylCh}) is the generalized Pauli channel $\Phi_{GPC}$ if and only if
\begin{equation}\label{GPC_group}
\ell_{\alpha k,\alpha l}=\ell_{kl},\qquad p_{\alpha k,\alpha l}=p_{kl}.
\end{equation}
One finds the following Kraus representation,
\begin{equation}\label{GPC_form}
\Phi_{GPC}[X]=\pi_0X+\frac{1}{d-1}\left[\sum_{k=1}^{d}\pi_k\sum_{\alpha=1}^{d-1}
W_{\alpha k,\alpha}XW_{\alpha k,\alpha}^\dagger+\pi_{d+1}\sum_{\alpha=1}^{d-1}W_{\alpha 0}XW_{\alpha 0}^\dagger\right],
\end{equation}
with $\pi_k\geq 0$, $\sum_{k=0}^{d+1}\pi_k=1$. Therefore, to obtain the generalized Pauli channels, we need to impose additional conditions on the quantum channels that are irreducibly covariant with respect to the group $G$ generated by the Weyl operators.

Let us recall that two unitary representations $V_\alpha$ and $V_\beta$ of the group $G$ are equivalent if and only if there exists the transformation matrix $S$ such that
\begin{equation}\label{equiv}
\bigforall_{g\in G}\quad V_\beta(g)=SV_\alpha(g)S^\dagger.
\end{equation}

\begin{Proposition}\label{prop2}
For the group generated by the Weyl operators, there are the following pairs of equivalent unitary representations,
\begin{equation}
\{U_\alpha,U_{d-\alpha}\},\quad\{\overline{U}_\alpha,\overline{U}_{d-\alpha}\},
\quad \alpha=1,\dots,d-1,
\end{equation}
where $S$ is the following permutation matrix
\begin{equation}\label{S}
S=\sum_{m=0}^{d-1}|m\>\<-m|.
\end{equation}
\end{Proposition}

\begin{proof}
The fact that the representations $U_\alpha$ and $U_{d-\alpha}$ (and also $\overline{U}_\alpha$ and $\overline{U}_{d-\alpha}$) are equivalent has been discussed in the proof to Proposition \ref{prop}. To find the transformation matrix $S$, we assume that
\begin{equation} \label{}
S=\sum_{m,n=0}^{d-1}s_{mn}W_{mn}.
\end{equation}
For the unitary representations $U_\alpha$ and $\overline{U}_\alpha$ of the group $G$, eq. (\ref{equiv}) implies that
\begin{equation}
\bigforall_{k,l=0,\dots,d-1}\quad W_{kl}=SW_{-k,-l}S^\dagger.
\end{equation}
Using the group action of $G$ generated by the Weyl operators $W_{kl}$, one gets (\ref{S}).
\end{proof}

\begin{Remark}
The transformation matrix $S$ from Proposition \ref{prop2} recovers $A_{00}$
in the definition of the discrete Wigner function \cite{Wootters},
\begin{equation}
\mathcal{W}_{kl}=\frac 1d \Tr(\rho A_{kl}),
\end{equation}
where
\begin{equation}\label{A}
A_{kl}=\sum_{m=0}^{d-1}\omega^{2(m-k)l}|m\>\<-m+2k|.
\end{equation}
\end{Remark}

\begin{Proposition}\label{prop3}
Let $S$ be the transformation matrix given in Proposition \ref{prop2}. If the Weyl channel $\Phi$ is covariant with respect to $S$, that is,
\begin{equation}\label{S_cov}
\Phi[SXS^\dagger]=S\Phi[X]S^\dagger,
\end{equation}
then it possesses real eigenvalues.
\end{Proposition}

\begin{proof}
Calculate the inverse relation between $p_{kl}$ and $\ell_{kl}$ from formula (\ref{p}),
\begin{equation}
\ell_{mn}=\sum_{k,l=0}^{d-1}\omega^{-mk+nl}p_{kl}.
\end{equation}
The eigenvalues $\ell_{mn}$ are real if and only if
\begin{equation}\label{real_EV}
\ell_{d-m,d-n}={\ell}_{mn}.
\end{equation}
For the Weyl operators and $S$ from Proposition \ref{prop2}, we have
\begin{equation}
SW_{mn}S^\dagger=W_{-m,-n}.
\end{equation}
Hence, we see that $\Phi$ satisfies eq. (\ref{S_cov}) if and only if
(\ref{real_EV}) holds.
\end{proof}

\begin{example}\label{ex}
For $d=3$, the Weyl channel $\Phi$ is given by its eigenvalues $\ell_{kl}$. Condition (\ref{real_EV}) implies that
\begin{equation}
\ell_{01}=\ell_{02},\quad \ell_{10}=\ell_{20},\quad \ell_{11}=\ell_{22},\quad \ell_{12}=\ell_{21}.
\end{equation}
Observe that the eigenvalues of $\Phi$ are grouped as in (\ref{GPC_group}). Therefore, eq. (\ref{S_cov}) provides the necessary and sufficient conditions for the three-dimensional Weyl channel $\Phi$ to be the generalized Pauli channel. However, for a prime $d>3$, conditions (\ref{S_cov}) are no longer sufficient.
\end{example}

\begin{Theorem}\label{Th}
Suppose that the Weyl channel $\Phi$ satisfies (\ref{S_cov}); that is, $\ell_{kl}$ are real. For every $\beta=1,\dots,\frac{d-1}{2}$, use the irreducible representation $U_\beta$ to construct the Weyl channel via
\begin{equation}\label{QCmu}
\Phi_\beta := \sum_{k,l=0}^{d-1}\ell_{kl}\Pi_{kl}^{(\beta)},
\end{equation}
where $\Pi_{kl}^{(\beta)}$ are the rank-1 projectors onto $W_{\beta k,\beta l}$.
Now, $\Phi$ is the generalized Pauli channel if and only if
\begin{equation}\label{th}
\Phi=\Phi_\beta ,
\end{equation}
for any ${\beta=1,\dots,\frac{d-1}{2}}$.
\end{Theorem}

\begin{proof}
Observe that $\Phi$ and $\Phi_\beta$ have the same spectrum but to different eigenvalues, as eq. (\ref{proj2}) produces
\begin{equation}\label{projectors}
\Pi_{kl}^{(\beta)}[X]=\frac{1}{d}W_{\beta k,\beta l}\Tr(W_{\beta k,\beta l}^\dagger X)=\Pi_{\beta m,\beta n}[X].
\end{equation}
Therefore, if $\Phi$ is the quantum channel, then so are $\Phi_\beta$.
Moreover, formula (\ref{th}), together with (\ref{projectors}), imposes the  constraints $\ell_{\beta k,\beta l}=\ell_{kl}$, which are consistent with condition (\ref{GPC_group}) for the generalized Pauli channels. Note that the assumptions of Theorem \ref{Th} can be modified. If we drop the requirement that the Weyl channel $\Phi$ has real eigenvalues but check condition (\ref{th}) for $\beta=1,\dots,d-1$, then we arrive at the same results.
\end{proof}

%{Sections 4 and 5 lead to analogical results if one constructs $\Phi$ from $\overline{U}_1$ and replaces $U_\beta$ with $\overline{U}_\beta$.}

\section{Class of covariant positive maps}

Consider the positive covariant linear map
\begin{equation}
\Phi[X] = \sum_{\alpha=0}^{d^2-1} q_\alpha W_\alpha X W_\alpha^\dagger,
\end{equation}
where $W_\alpha$ are the Weyl operators. Recall that a map is positive if for any $X \geq 0$, one has $\Phi[X] \geq 0$ \cite{Paulsen,Bhatia}. Positive maps have received considerable attention recently due to their close relation to the entanglement theory \cite{HHHH}. Recall that a quantum state represented by a density operator $\rho \in \mathcal{B}(\mathcal{H}_1 \ot \mathcal{H}_2)$ is separable if and only if $(\oper_1 \ot \Phi)[\rho] \geq 0$ for all positive maps $\Phi : \mathcal{L}(\mathcal{H}_2) \longrightarrow \mathcal{L}(\mathcal{H}_1)$. Unfortunately, the structure of positive maps is still not fully understood, and the general construction of such maps is not known (see the recent review \cite{KYE,TOPICAL}).

In what follows, we provide two constructions of the positive maps that are covariant with respect to the finite group generated by the Weyl operators. Let $\Delta \subset \{0,1,\ldots,d^2-1\}$ with $|\Delta|=N$. Consider the linear map $\Phi : M_d(\mathbb{C}) \longrightarrow  M_d(\mathbb{C})$ given by
\begin{equation}\label{PI}
  \Phi[X] = \sum_{\alpha \notin  \Delta} \lambda_\alpha^+ F_\alpha X F_\alpha^\dagger +  \sum_{\alpha \in \Delta} \lambda_\alpha^- F_\alpha X F_\alpha^\dagger ,
\end{equation}
with $\lambda_\alpha^+ >0$ and $\lambda_\alpha^- < 0$. The operators $F_\alpha$ define an orthonormal basis in $ M_d(\mathbb{C})$,
\begin{equation}\label{}
  \Tr( F_\alpha^\dagger F_\beta) = \delta_{\alpha\beta} \ ;\quad\alpha,\beta=0,1,\ldots,d^2-1\ .
\end{equation}
Clearly, if there are no negative eigenvalues, the map is completely positive. Now, the problem is to find the `balance' between $\lambda_\alpha^+$ and $\lambda_\alpha^- $ that guarantees the positivity of $\Phi$. Denote the operator norm (the largest singular value) of $A$ by $||A||$.

\begin{Theorem}[\cite{CMP}] If $\sum_{\alpha \in \Delta} || F_\alpha||^2 < 1$ and
\begin{equation}\label{+-}
  \lambda_\alpha^+ \geq \frac{ \sum_{\beta \in \Delta} | \lambda_\beta^-| || F_\beta||^2 }{1- \sum_{\beta \in \Delta} || F_\beta||^2 } \ ; \ \ \alpha \notin \Delta \ ,
\end{equation}
then $\Phi$ defined by (\ref{PI}) is positive.
\end{Theorem}

Clearly, conditions (\ref{+-}) are only sufficient but not necessary. Now, taking $F_\alpha = W_\alpha/\sqrt{d}$, one has $|| F_\alpha||^2 = 1/d$, and hence one arrives at the following corollary.

\begin{Corollary} If
\begin{equation}\label{}
  \lambda_\alpha^+ \geq \frac{1}{d-N} \sum_{\beta \in \Delta} | \lambda_\beta^-| \ ; \ \ \alpha \notin \Delta \ ,
\end{equation}
for any $N \leq d-1$ and $F_\alpha = W_\alpha/\sqrt{d}$,
then formula (\ref{PI}) defines a covariant positive map.
\end{Corollary}

The above construction of covariant positive maps allows at most `$d-1$' negative $\lambda_\alpha$.

\begin{example}
For $d=2$, let us take $\lambda^-_0=-1$. Now, if $\lambda_\alpha^+ \geq 1$, then the map
\begin{equation}\label{}
  \Phi[X] = \frac 12 \left( \sum_{\alpha=1}^3  \lambda_\alpha^+ \sigma_\alpha X \sigma_\alpha - X \right)
\end{equation}
is positive. In particular, taking $ \lambda_\alpha^+ = 1$ and using the well-known property
\begin{equation}\label{}
  \sum_{\alpha=1}^3   \sigma_\alpha X \sigma_\alpha = 2 \mathbb{I}_2 \, \Tr X - X,
\end{equation}
one recovers the celebrated reduction map
\begin{equation}\label{}
 \Phi[X] = \mathbb{I}_2 \Tr X - X.
\end{equation}
\end{example}

\begin{example}
The above example can be generalized to an arbitrary dimension in the following way: take $\lambda_\alpha^-=-1$ for $\alpha=0,1,\ldots,d-2$, and $\lambda_\alpha^+=d-1$ for $\alpha=d-1,\ldots,d^2-1$. Then, the map
\begin{equation}\label{}
  \Phi[X]= \frac{1}{d(d-1)^2} \left\{ (d-1)\sum_{\alpha=d-1}^{d^2-1} W_\alpha X W_\alpha^\dagger -  \sum_{\alpha=0}^{d-2} W_\alpha X W_\alpha^\dagger \right\}
\end{equation}
is positive and trace-preserving.
\end{example}

Let us turn our interest to prime dimensions. By $\{ |\psi^{(\alpha)}_1\>,\ldots,|\psi^{(\alpha)}_{d}\> \}$ with $\alpha=1,\ldots,d+1$, denote $d+1$ mutually unbiased bases (MUBs) in $\mathbb{C}^d$; that is, the orthonormal bases for which
\begin{equation}\label{}
  | \< \psi^{(\alpha)}_k| \psi^{(\beta)}_l\>|^2 = \frac 1d
\end{equation}
when $\alpha \neq \beta$. In \cite{MUBs}, the authors considered the following family of positive, trace-preserving maps,
\begin{equation}\label{Phi-MUB}
\Phi[X]=\frac{1}{d-1}\left[2\mathbb{I}\Tr X-\sum_{\alpha=1}^{d+1}
\sum_{k,l=0}^{d-1} \mathcal{O}_{kl}^{(\alpha)} \Tr(XP_l^{(\alpha)})P_k^{(\alpha)}\right].
\end{equation}
In the above formula, $\mathcal{O}^{(\alpha)}$ are rotations in $\mathbb{R}^d$ around the axis determined by $\mathbf{n}=(1,\ldots,1)$; i.e., $\mathcal{O}^{(\alpha)} \mathbf{n}=\mathbf{n}$.

\begin{Proposition}
The map (\ref{Phi-MUB}) is Weyl-covariant if and only if $\mathcal{O}^{(\alpha)} = \mathbb{I}_d$ for $\alpha=1,\ldots,d+1$.
\end{Proposition}

\begin{proof}
One can construct unitary operators from MUBs in the following way,
\begin{equation}
U_\alpha=\sum_{l=1}^d\omega^lP_l^{(\alpha)}.
\end{equation}
Recall that, for prime $d$, $U_\alpha^k$ are the rescalled Weyl operators \cite{DCKS}. {From the properties of the Weyl channels,} it follows that $\Phi$ is Weyl-covariant if and only if $\mathbb{I}$ and $U_\alpha^k$ are its eigenvectors. We check that
\begin{equation}\label{ev}
\begin{split}
\Phi[U_\beta^m]
%=&-\frac{1}{d-1}\sum_{\alpha=1}^{d+1}
%\sum_{k,l,n=1}^d\mathcal{O}_{kl}^{(\alpha)}\omega^{nm}\Tr(P_n^{(\beta)}P_l^{(\alpha)})P_k^{(\alpha)}
=-\frac{1}{d-1}\sum_{\alpha=1}^{d+1}\sum_{k,l,n=1}^d\mathcal{O}_{kl}^{(\alpha)}\omega^{nm}P_k^{(\alpha)}
\left(\delta_{nl}\delta_{\alpha\beta}+\frac 1d (1-\delta_{\alpha\beta})\right)
%\\=&-\frac{1}{d-1}\left(\sum_{k,l=1}^d\mathcal{O}_{kl}^{(\beta)}\omega^{lm}P_k^{(\beta)}
%+\frac 1d \sum_{\alpha=1}^{d+1}\sum_{k,l,n=1}^d\mathcal{O}_{kl}^{(\alpha)}
%\omega^{nm}P_k^{(\alpha)}-\frac 1d %\sum_{k,l,n=1}^d\mathcal{O}_{kl}^{(\beta)}\omega^{nm}P_k^{(\beta)}\right)\\
%=&
=-\frac{1}{d-1}\sum_{k,l=1}^d\mathcal{O}_{kl}^{(\beta)}\omega^{lm}P_k^{(\beta)},
\end{split}
\end{equation}
and hence $U_\alpha^k$ are the eigenvectors of $\Phi$ if and only if
\begin{equation}\label{O}
\mathcal{O}^{(\alpha)}=\mathbb{I}_d.
\end{equation}
\end{proof}

\begin{Remark}
If $\mathcal{O}^{(\alpha)} = \mathbb{I}_d$ for $\alpha=1,\ldots,d+1$, then map (\ref{Phi-MUB}) reduces to
\begin{equation}\label{red}
\Phi[X]= \frac{1}{d-1}\left( 2 \mathbb{I}\, \Tr X - \sum_{\alpha=1}^{d+1} \Phi_\alpha[X] \right),
\end{equation}
where
\begin{equation}\label{}
\Phi_\alpha[X] = \sum_{k=0}^{d-1} P_k^{(\alpha)} X P_k^{(\alpha)} .
\end{equation}
Knowing that
\begin{equation}\label{}
  \sum_{\alpha=1}^{d+1} \Phi_\alpha = \oper + d\Phi_0 ,
\end{equation}
with $\Phi_0[X]= \frac{1}{d}\mathbb{I}\Tr X$ being the completely depolarizing channel, we recover the reduction map in $M_d(\mathbb{C})$,
\begin{equation}\label{}
\Phi[X]= \frac{1}{d-1}\left( \mathbb{I}\, \Tr X - X \right).
\end{equation}
\end{Remark}

Now, let us generalize the map in (\ref{red}) as follows.

\begin{Proposition}
Take a subset $\Gamma \subset \{1,2,\ldots,d+1\}$ with $|\Gamma|=k$. The map defined by
\begin{equation}\label{gen}
\Phi_\Gamma := \frac{1}{d-1}\left( 2(k-1) \Phi_0 + \sum_{\alpha \notin \Gamma} \Phi_\alpha - \sum_{\alpha \in \Gamma} \Phi_\alpha \right)
\end{equation}
is the generalized Pauli trace-preserving positive map.
\end{Proposition}

\begin{proof}
Let $P$ be any rank-1 projector. We show that
\begin{equation}\label{}
\Tr(\Phi_\Gamma[P])^2 = \frac{1}{d-1},
\end{equation}
which implies the positivity of $\Phi_\Gamma$ (cf. \cite{MUBs}). We consider two separate cases:

\begin{enumerate}

\item $k=1$,

\item $k=2,\ldots,d+1$.

\end{enumerate}

If $k=1$, then $\Gamma=\{ \alpha_*\}$, and $\Phi_\Gamma$ simplifies to
\begin{equation}\label{gen-1}
  \Phi_\Gamma := \frac{1}{d-1}\left(  {\sum_{\alpha}}^\prime \Phi_\alpha -  \Phi_{\alpha_*} \right) ,
\end{equation}
where ${\sum}^\prime$ denotes the sum over $\alpha \neq \alpha_*$. We need to find
\begin{eqnarray}
\Tr(\Phi_\Gamma[P])^2 = \frac{1}{(d-1)^2} \left\{
\sum_{\alpha=1}^{d+1} \Tr(\Phi_\alpha[P])^2 +{\sum_{\alpha \neq \beta}}^\prime  \Tr(\Phi_\alpha[P] \Phi_\beta[P] )  - 2 {\sum_{\alpha}}^\prime \Tr( \Phi_\alpha[P] \Phi_{\alpha_*}[P]) \right\}.
\end{eqnarray}
Observe that
\begin{equation}\label{}
  \Tr( \Phi_\alpha[P] \Phi_\alpha[P] ) =  \sum_{k=0}^{d-1}  \sum_{l=0}^{d-1} \Tr\left( P_k^{(\alpha)} P P_k^{(\alpha)}   P_l^{(\alpha)} P P_l^{(\alpha)} \right) =  \sum_{k=0}^{d-1} \Tr(  P P_k^{(\alpha)} )^2,
\end{equation}
and hence
\begin{equation}\label{a1}
\sum_{\alpha=1}^{d+1} \Tr(\Phi_\alpha[P])^2 = \sum_{\alpha=1}^{d+1} \sum_{k=0}^{d-1}  \Tr(  P P_k^{(\alpha)} )^2  = 2 ,
\end{equation}
where we have used the well-known property of the MUBs \cite{Wu,Beatrix}. Moreover, using the fact that all the maps $\Phi_\alpha$ are self-dual and $\Phi_\alpha \Phi_\beta = \Phi_0$ for $\alpha \neq \beta$, one finds
\begin{equation}\label{a2}
\Tr(\Phi_\alpha[P] \Phi_\beta[P] )  =  \Tr(P \Phi_\alpha[\Phi_\beta[P]] ) = \Tr(P \Phi_0[P] ) = \frac 1d .
\end{equation}
Finally, we arrive at
\begin{eqnarray}
\Tr(\Phi_\Gamma[P])^2 = \frac{1}{(d-1)^2} \left[ 2 + (d-1) - 2 \right] = \frac{1}{d-1} .
\end{eqnarray}

For $k>1$, {we arrive at}
\begin{eqnarray}
  \Tr(\Phi_\Gamma[P])^2 &=& \frac{1}{(d-1)^2} \Tr \left\{ 4(k-1)^2 (\Phi_0[P])^2 + %\sum_{\alpha=1}^{d+1} (\Phi_\alpha[P])^2
   4(k-1) \Phi_0[P] \left( \sum_{\alpha \notin \Gamma} \Phi_\alpha[X] - \sum_{\alpha \in \Gamma} \Phi_\alpha[X] \right)    \right.  \\
   &+&  \left.  \sum_{\alpha \notin \Gamma} \Phi_\alpha[X] \sum_{\beta \notin \Gamma} \Phi_\beta[X] +  \sum_{\alpha \in \Gamma} \Phi_\alpha[X] \sum_{\beta \in \Gamma} \Phi_\beta[X]      - \sum_{\alpha \notin \Gamma} \Phi_\alpha[X]  \sum_{\beta \in \Gamma} \Phi_\beta[X]    -    \sum_{\beta \in \Gamma} \Phi_\beta[X]  \sum_{\alpha \notin \Gamma} \Phi_\alpha[X] \right\} , \nonumber
\end{eqnarray}
which simplifies to
\begin{equation}
\begin{split}
\Tr(\Phi_\Gamma[P])^2\ =&\ \ \frac{1}{(d-1)^2} \left\{ \frac{4(k-1)^2}{d}  + %\sum_{\alpha=1}^{d+1} (\Phi_\alpha[P])^2
   \frac{4(k-1)}{d} \left( d+1-2k \right)   +  \sum_{\alpha=1}^{d+1} \Tr(\Phi_\alpha[P])^2  \right.  \\
+&  \left.  \sum_{\alpha\neq \beta \notin \Gamma} \Tr(\Phi_\alpha[P]  \Phi_\beta[P] )+  \sum_{\alpha \neq \beta \in \Gamma} \Tr(\Phi_\alpha[P]  \Phi_\beta[P] )     - 2\sum_{\alpha \notin \Gamma} \sum_{\beta \in \Gamma}  \Tr(\Phi_\alpha[P] \Phi_\beta[P] )   \right\} ,
\end{split}
\end{equation}
Finally, using eqs. (\ref{a1}-\ref{a2}), one gets
\begin{equation}
\begin{split}
  \Tr(\Phi_\Gamma[P])^2\ &=\ \ \frac{1}{d(d-1)^2} \left\{ {4(k-1)^2}  + %\sum_{\alpha=1}^{d+1} (\Phi_\alpha[P])^2
   {4(k-1)} \left( d+1-2k \right)   +  2d  +2 \left( \begin{array}{c}  d+1-k \\ 2 \end{array} \right) \right.  \\
   &+  \left.    2 \left( \begin{array}{c} k \\ 2 \end{array} \right)     - 2 k (d+1-k)   \right\} =   \frac{1}{d-1} ,
\end{split}
\end{equation}
which ends the proof.
\end{proof}

%$$   ================================== $$
%$$   ================================== $$
%
%Let us observe that
%\begin{equation}\label{}
%  \Phi_0 \Phi_0 = \Phi_0 \ , \ \  \Phi_0 \Phi_\alpha =  \Phi_\alpha \Phi_0 = \Phi_0 \ , \ \ \Phi_\alpha \Phi_\alpha = \Phi_\alpha ,
%\end{equation}
%and
%\begin{equation}\label{}
%  \Phi_\alpha \Phi_\beta = \Phi_0 \ , \ \ \alpha \neq \beta ,
%\end{equation}
%and hence using the fact that all maps $\Phi_\alpha$ are self-dual one finds
%
%
%
%$$ =================================================  $$
%
%the maximally mixed probability vector. We are interested in the subclass of $d$-dimensional $\Phi_P$ with prime $d$, consisting in the Weyl-covariant linear maps.
%
%
%
%  Let us denote the corresponding rank-1 projectors by $P^{(\alpha)}_l = |\psi^{(\alpha)}_l\>\< \psi^{(\alpha)}_l|$.
%
%
%
%
%
%\begin{Theorem}
%The positive linear map $\Phi_P$ defined in (\ref{phi}) is Weyl-covariant if and only if it has the form of the generalized Pauli channel (\ref{GPC_form}) but with $p_\alpha\ngeq 0$.
%\end{Theorem}

\section{Conclusions}

{
We analyzed a class of irreducibly covariant quantum channels with respect to the unitary representations of the finite group $G$ generated by the Weyl operators. Interestingly, all the $d$-dimensional representations of $G$ are multiplicity-free, and hence one can apply to them the analysis developed recently by Mozrzymas et. al. \cite{MSD}. The channel that is irreducibly covariant with respect to the $d$-dimensional representations of $G$ turns out to be the Weyl channel. For prime dimensions, one can consider the subclass of the Weyl channels known as the generalized Pauli channels \cite{Ruskai,DCKS,DCKS2}. Finally, we analyzed the class of irreducibly covariant maps which are positive but not necessarily completely positive. Such maps provide an important tool in the entanglement theory. We presented two classes of positive maps: one belonging to the general family of the Weyl maps and the other to the reduced family of the generalized Pauli maps. It would be interesting to further study the properties of positive maps -- for example, decomposability, optimality, or extremality. Also, the generalized Pauli maps were defined only for prime dimensions. One could consider a more general case with $d=p^r$ and prime $p$, where the maximal set of `$d+1$' MUBs can be found \cite{Wootters}.}

\section*{Acknowledgements} This paper was partially supported by the National Science Centre project 2015/19/B/ST1/03095.

\end{document}